\documentclass[preprint,12pt]{elsarticle}

\usepackage{bm,paralist,graphicx,braket}

\usepackage{caption,subcaption}
\usepackage{amssymb,amsmath,amsthm}
\theoremstyle{plain}
\newtheorem{theorem}{Theorem}
\newtheorem{proposition}[theorem]{Proposition}
\newtheorem{corollary}[theorem]{Corollary}
\newtheorem{remark}[theorem]{Remark}

\usepackage[nameinlink]{cleveref}

\newcommand\od[2]{\frac{d#1}{d#2}}
\newcommand\pd[2]{\frac{\partial #1}{\partial #2}}

\newcommand\tpd[2]{\partial #1/\partial #2}

\newcommand\parentheses[1]{\!\left(#1\right)}

\newcommand\norm[1]{\left\|#1\right\|}
\newcommand\abs[1]{\left|#1\right|}
\newcommand\DS{\displaystyle}
\newcommand\R{\mathbb{R}}

\newcommand\defeq{\mathrel{\mathop:}=}

\newcommand\eps{\varepsilon}

\newenvironment{tbmatrix}{\left[\begin{smallmatrix}}{\end{smallmatrix}\right]}

\newcommand\bd{\mathbf{d}}

\newcommand\PB[2]{\left\{#1,#2\right\}}

\newcommand{\SO}{\mathsf{SO}}

\newcommand{\br}{\bm{r}}
\newcommand{\bp}{\bm{p}}
\newcommand{\bz}{\bm{z}}

\newcommand{\ez}{\bm{e}_{z}}
\newcommand\dt{\Delta t}

\journal{Journal of Computational Physics}

\begin{document}

\begin{frontmatter}



\title{Explicit Symplectic Integrators for Massive Point Vortex Dynamics in Binary Mixture of Bose--Einstein Condensates}


\author{Tomoki Ohsawa} 

\affiliation{organization={Department of Mathematical Sciences, The University of Texas at Dallas},
            addressline={800 W Campbell Rd}, 
            city={Richardson},
            postcode={75080-3021},
            state={Texas},
            country={United States}}

\begin{abstract}
We construct explicit integrators of arbitrary even orders of accuracy for massive point vortex dynamics in binary mixture of Bose--Einstein condensates proposed by Richaud et al. The integrators are symplectic and preserve the angular momentum of the system exactly. Our main focus is the small-mass regime in which the minor component of the binary mixture comprises a very small fraction of the total mass. The solution behaviors in this regime change significantly depending on the initial momenta: they are highly oscillatory unless the momenta satisfy certain conditions. The standard Runge--Kutta method performs very poorly in preserving the Hamiltonian showing a significant drift in the long run, especially for highly oscillatory solutions. On the other hand, our integrators nearly preserve the Hamiltonian without drifts. We also give an estimate of the error in the Hamiltonian by finding an asymptotic expansion of the modified Hamiltonian for our second-order integrator.
\end{abstract}



\end{frontmatter}

\section{Massive Point Vortex Dynamics}
\subsection{Massive Point Vortices in Two-Component BEC}
The main focus of this paper is to numerically solve the equations of motion for $N$ massive point vortices in a pancake-shaped Bose--Einstein condensate (BEC) of topological charges $\{ q_{j} = \pm 1 \}_{j=1}^{N}$ located at $\{ \br_{j} \defeq (x_{j}, y_{j}) \in \R^{2} \}_{j=1}^{N}$.

We note in passing that the term ``massive'' is not in the sense that it is great in size or number of vortices but rather that the vortex behaves as if it possesses (usually small) mass unlike the more conventional point-vortex model that assumes that the vortex possesses no mass.
In other words, massive vortices are described by first-order differential equations in the phase space (position--momentum space), whereas massless vortices are described by first-order equations in the configuration space (position space).

We set $r_{j} \defeq |\br_{j}| = (\br_{j} \cdot \br_{j})^{1/2}$ to be the length of $\br_{j}$, and also use shorthands $\br \defeq (\br_{1}, \dots, \br_{N})$ and similarly for other vectors.
We also set $\ez \defeq (0,0,1)$ and note that, for every pair of $\bm{a}, \bm{b} \in \R^{2}$, the cross product $\bm{a} \times \bm{b}$ are taken by attaching zero as the third components to both, and we see the result as a vector in $\R^{2}$ or $\R^{3}$ depending on the context.

The Lagrangian (in the non-dimensional form) for the massive $N$ vortices in a binary mixture of BEC with components $a$ and $b$ is given by (see \cite{PhysRevA.101.013630})
\begin{equation}
  \label{eq:L}
  L(\br,\dot{\br}) \defeq \sum_{j=1}^{N} \parentheses{
    \frac{\eps}{2} \dot{\br}_{j}^{2}
    + q_{j} (\dot{\br}_{j} \times \br_{j}) \cdot \ez
  }
  - E(\br),
\end{equation}
where the parameter $\eps$ is defined as
\begin{equation}
  \label{eq:eps}
  \eps \defeq \frac{M_{b}/M_{a}}{N},
\end{equation}
where $M_{s}$ with $s = a, b$ is the total mass of the components/species $s$; the potential term $E$ is given by
\begin{equation}
  \label{eq:E}
  E(\br)
  \defeq \sum_{j=1}^{N} \ln(1 - r_{j}^{2})
  + \sum_{1\le j < k \le N} q_{j} q_{k} \ln\parentheses{ \frac{1 - 2 \br_{j} \cdot \br_{k} + r_{j}^{2} r_{k}^{2}}{|\br_{j} - \br_{k}|^{2}} },
\end{equation}
where the first term comes from a confinement to the unit disc on the plane and the second term from interactions of the $N$ vortices.
The Euler--Lagrange equation then gives
\begin{equation}
  \label{eq:Euler-Lagrange}
  \eps\,\ddot{\br}_{j} + 2 q_{j} J \dot{\br}_{j} = - \nabla_{j}E(\br),
\end{equation}
where $\nabla_{j} \defeq \tpd{}{\br_{j}}$.

The Lagrangian~\eqref{eq:L} was derived by \cite{PhysRevA.101.013630} by a variational approximation of a two-component Gross--Pitaevskii (GP) equations for a binary mixture of BECs.
This was motivated by their earlier work~\cite{PhysRevA.103.023311} using a coupled GP equations for such a binary mixture in the immiscible regime.
Specifically, solutions of the coupled GP equations show that the majority component exhibits vortices, and the atoms of the minority component are trapped inside the vortices.
This results in equipping the vortices with masses, in contrast to the standard quantum vortices~\cite{Fe1955,On1949} that are usually considered to be massless, and is often approximated by the Kirchhoff equations (see \eqref{eq:Kirchhoff} below).

The variational approximation in \cite{PhysRevA.101.013630} assumes, for the major (massless) $a$-species, the ansatz in the form of the trial wave function from \cite{PhysRevA.70.043624} for $N$ vortices located at $\{ \br_{j} \}_{j=1}^{N}$, whereas it assumes, for the minor $b$-species, a linear combination of Gaussians from \cite{PeMiCiLeZo1996} centered at $\{ \br_{j} \}_{j=1}^{N}$ as well.

We are particularly interested in the regime where $\eps \ll 1$, that is, the $b$-species comprise a small mass compared to the $a$-species, but its presence is not negligible.
One sees that then \eqref{eq:Euler-Lagrange} is a singularly perturbed system.

\subsection{Hamiltonian Formulation}
Using the Lagrangian~\eqref{eq:L}, the Legendre transformation is defined via the momenta $\bp \defeq (\bp_{1}, \dots, \bp_{N})$ with
\begin{equation}
  \label{eq:p}
  \bp_{j} \defeq \pd{L}{\dot{\br}_{j}}
  = \eps\,\dot{\br}_{j} + q_{j}(\br_{j} \times \ez)
  = \eps\,\dot{\br}_{j} + q_{j} J \br_{j}, 
\end{equation}
where we set
\begin{equation}
  \label{eq:J}
  J \defeq
  \begin{bmatrix}
    0 & 1 \\
    -1 & 0
  \end{bmatrix}
  \text{ so that }
  J \bm{a} = \bm{a} \times \ez
  \quad
  \forall \bm{a} \in \R^{2}.
\end{equation}
Hence we have $\dot{\br}_{j} = \frac{1}{\eps} ( \bp_{j} - q_{j} J \br_{j} )$, and so have the Hamiltonian
\begin{align}
  \label{eq:H}
  H(\br,\bp)
  &\defeq \sum_{j=1}^{N} \bp_{j} \cdot \dot{\br}_{j} - L(\br,\dot{\br}) \nonumber\\
  &= \frac{1}{2\eps} \sum_{j=1}^{N} \parentheses{ \bp_{j} - q_{j} J \br_{j} }^{2} + E(\br). 
\end{align}

Notice that the Hamiltonian is not separable, i.e., $H(\br,\bp) \neq T(\bp) + V(\br)$ with some functions $T$ and $V$.
It is well known that there is no \textit{explicit} symplectic integrator for general non-separable Hamiltonian systems~\cite{SaCa2018,LeRe2004,HaLuWa2006}.

Let us set
\begin{equation*}
  \bz_{j} =
  \begin{bmatrix}
    \br_{j} \\
    \bp_{j}
  \end{bmatrix},
  \quad
  \bz = (\bz_{1}, \dots, \bz_{N}),
  \quad
  \mathbb{J}_{n} \defeq
  \begin{bmatrix}
    0 & I_{n} \\
    -I_{n} & 0
  \end{bmatrix},
\end{equation*}
where $I_{n}$ is the $n \times n$ identity matrix, and consider Hamilton's equations
\begin{equation}
  \label{eq:Hamilton0}
  \dot{\bz} = \mathbb{J}_{2N}\, \nabla H(\bz)
  \iff
  \left\{
  \begin{array}{l}
    \DS\dot{\br}_{j} = \pd{H}{\bp_{j}},
    \medskip\\
    \DS\dot{\bp}_{j} = -\pd{H}{\br_{j}},
  \end{array}
  \right.
\end{equation}
where $j = 1, \dots, N$, or more concretely,
\begin{equation}
  \label{eq:Hamilton}
  \begin{array}{l}
  \DS\dot{\br}_{j} = \frac{1}{\eps}(- q_{j} J \br_{j} + \bp_{j}),
    \medskip\\
  \DS\dot{\bp}_{j} = \frac{1}{\eps}(-\br_{j} - q_{j} J \bp_{j}) - \nabla_{j}E(\br),
  \end{array}
\end{equation}
noting that $q_{j} = \pm 1$.

\subsection{Symplecticity and Noether Invariant}
Since each vortex is constrained to the open unit disk
\begin{equation*}
  \mathcal{D} \defeq \Set{ \bm{x} \in \R^{2} |\, |\bm{x}| < 1 },
\end{equation*}
the phase space for the Hamiltonian system~\eqref{eq:Hamilton} is 
\begin{equation}
  \mathcal{P} \defeq \Set{ \bz = (\br,\bp) \in \mathcal{D}^{N} \times \R^{2N} | \bp \in \R^{2N} },
\end{equation}
which is equipped with the standard symplectic form
\begin{equation}
  \label{eq:Omega}
  \Omega \defeq \bd\br_{j} \wedge \bd\bp_{j} = \bd{x}_{j} \wedge \bd{\xi_{j}} + \bd{y}_{j} \wedge \bd{\eta_{j}},
\end{equation}
where $\bd$ stands for the exterior derivative, $\bp_{j} = (\xi_{j}, \eta_{j})$, and the summation convention is assumed on $j$.

Let $\Phi_{t}$ be the flow of \eqref{eq:Hamilton}, i.e., for every $t \in \R$ for which the solution $\bz(t) = (\br(t), \bp(t))$ exists with initial point $\bz(0) = (\br(0), \bp(0))$,
\begin{equation*}
  \Phi_{t}(\bz(0)) = \bz(t).
\end{equation*}
Then $\Phi_{t}$ is symplectic, i.e.,
\begin{equation*}
  \Phi_{t}^{*} \Omega = \Omega
  \iff
  D\Phi_{t}(\bz)^{T} \mathbb{J}_{2N} D\Phi_{t}(\bz) = \mathbb{J}_{2N},
\end{equation*}
where $D\Phi_{t}$ stands for the Jacobian matrix of $\Phi_{t}$ with respect to the variables $\bz = (\br,\bp)$.

One observes that the Hamiltonian~\eqref{eq:H} possesses the (planar) rotational symmetry: 
\begin{multline}
  \label{eq:SO2-sym}
  H\parentheses{ R\br_{1}, \dots, R\br_{N}, R\bp_{1}, \dots, R\bp_{N} }
  \\
  = H(\br_{1}, \dots, \br_{N}, \bp_{1}, \dots, \bp_{N})
  \quad\forall R \in \SO(2).
\end{multline}
As a result, the total angular momentum
\begin{equation}
  \label{eq:ell}
  \ell(\br,\bp) \defeq \sum_{j=1}^{N} (\br_{j} \times \bp_{j}) \cdot \ez
\end{equation}
gives the corresponding Noether invariant, and is conserved by \eqref{eq:Hamilton}.

\subsection{Oscillatory Solutions and Separation of Scales}
The solutions of \eqref{eq:Hamilton} tend to be highly oscillatory when $\eps \ll 1$.
It was also found in our recent work~\cite{OhRiGo-small_mass1} that the initial point $(\br(0),\bp(0))$ may affect the oscillatory nature of the solution.
Specifically, consider the subset
\begin{equation*}
  \mathcal{K} \defeq \Set{ \bz = (\br,\bp) \in \mathcal{P} | \bp_{j} = q_{j} J \br_{j} \text{ for } 1\le j \le N  }.
\end{equation*}
Notice that the Hamiltonian $H(\br,\bp)$ (see \eqref{eq:H}) of the massive dynamics restricted to $\mathcal{K}$ gives $E(\br)$, but then this is the Hamiltonian for the massless dynamics or the Kirchhoff equations:
\begin{equation}
  \label{eq:Kirchhoff}
  2q_{j} \dot{x}_{j} = \pd{E}{y_{j}},
  \qquad
  2q_{j} \dot{y}_{j} = -\pd{E}{x_{j}},
\end{equation}
which follows from the Euler--Lagrange equation~\eqref{eq:Euler-Lagrange} by taking the limit $\eps \to 0$.

It was proved in \cite{OhRiGo-small_mass1} that the massive dynamics---solutions of \eqref{eq:Hamilton}---with $\bz(0) \in \mathcal{K}$ stays $O(\eps)$-close to $\mathcal{K}$ for short times.
It was also observed numerically in \cite{OhRiGo-small_mass1} that the massive dynamics with $\bz(0) \notin \mathcal{K}$ exhibits fast oscillations with characteristic time of scale $O(\eps)$, whereas if $\bz(0) \in \mathcal{K}$ then such oscillations subside and the massive dynamics behaves like the massless dynamics~\eqref{eq:Kirchhoff} with characteristic time of scale $O(1)$.
Indeed, one can also prove that the massive dynamics with $\bz(0) \in \mathcal{K}$ stays $O(\eps)$-close to the corresponding massless dynamics---solutions of \eqref{eq:Kirchhoff} with same initial positions $\br(0)$.
Hence $\mathcal{K}$ is called the \textit{kinematic subspace} in \cite{OhRiGo-small_mass1} in the sense that this is a domain in the phase space $\mathcal{P}$ where the massive dynamics effectively loses its mass/inertia and hence the dynamics becomes more massless/kinematic.

Intuitively, the highly oscillatory behaviors come from the kinetic-energy (or inertia) terms in $H(\br,\bp)$ that are proportional to $1/\eps$.
These terms vanish on $\mathcal{K}$ and hence the fast (oscillatory) dynamics becomes less prominent near $\mathcal{K}$; as a result, the dynamics is dominated by the slow dynamics~\eqref{eq:Kirchhoff} driven by $E(\br)$.

Such a separation of scales in ordinary differential equations (ODEs) poses a stiff problem---a class of ODEs that are challenging to solve numerically because of a disparity in the time scales of the rapid transient behaviors and the slower global behaviors~\cite{HaNoWa1993b}.

\subsection{Illustrative Example: Single Massive Vortex}
\label{ssec:single_vortex}
In order to illustrate the characteristics of the system~\eqref{eq:Hamilton} described above, let us consider a simple example of a single massive vortex ($N = 1$) with charge $q_{1} = 1$ and $\eps = 0.01$.
Note that, although the interaction terms in $E$ are absent, there is still the confinement term---the first term on the right-hand side in \eqref{eq:E}---in this system.
As a result, one expects to observe fast oscillations and the separation of time scales described above.

Let us consider the initial condition
\begin{equation}
  \label{eq:singlevort_IC_in_K}
  \br(0) =
  \begin{bmatrix}
    0.5 \\
    0.3
  \end{bmatrix}
  \qquad
  \bp(0) =
  \begin{bmatrix}
    0.3 \\
    -0.5
  \end{bmatrix}
\end{equation}
that satisfies $(\br(0), \bp(0)) \in \mathcal{K}$, as well as the one with the same $\br(0)$ from above but with the second component of $\bp(0)$ from above reversed:
\begin{equation}
  \label{eq:singlevort_IC_off_K}
  \br(0) =
  \begin{bmatrix}
    0.5 \\
    0.3
  \end{bmatrix}
  \qquad
  \bp(0) =
  \begin{bmatrix}
    0.3 \\
    0.5
  \end{bmatrix},
\end{equation}
for which $(\br(0), \bp(0)) \notin \mathcal{K}$.

\begin{figure}[hbtp]
  \centering
    \begin{subcaptionblock}{.49\linewidth}
    \centering
    \includegraphics[width=\linewidth]{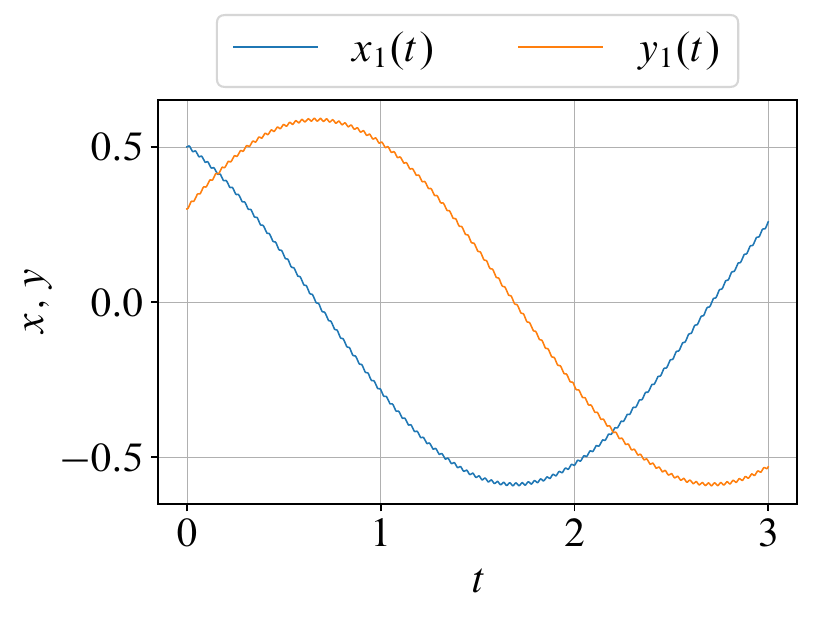}
    \caption{Initial condition~\eqref{eq:singlevort_IC_in_K} with $(\br(0), \bp(0)) \in \mathcal{K}$}
    \label{fig:t-xy_inK_001}
  \end{subcaptionblock}
  \begin{subcaptionblock}{.49\linewidth}
    \centering
    \includegraphics[width=\linewidth]{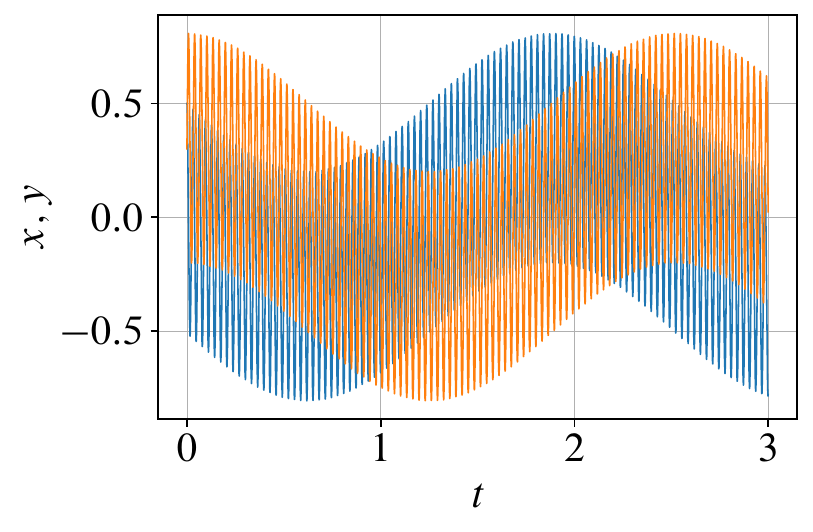}
    \caption{Initial condition~\eqref{eq:singlevort_IC_off_K} with $(\br(0), \bp(0)) \notin \mathcal{K}$}
    \label{fig:t-xy_offK_001}
  \end{subcaptionblock}
  \caption{Time evolution of $(x,y)$-coordinates of single massive vortex: $N = 1$, $q_{1} = 1$, $\eps = 0.01$.}
  \label{fig:t-xy_001} 
\end{figure}

\Cref{fig:t-xy_001} shows the time evolution of the $(x,y)$-coordinates of the single massive vortex, computed by the 4th-order symplectic method we shall construct below; see \eqref{eq:Strang} and \eqref{eq:nth-order-tripleJump} below with $n = 4$.
The solution with $(\br(0), \bp(0)) \in \mathcal{K}$ exhibits only small fluctuations that are barely visible on the plot in panel~\subref{fig:t-xy_inK_001}, and seems to be dominated by the slow dynamics.
On the other hand, the solution with $(\br(0), \bp(0)) \notin \mathcal{K}$ in panel~\subref{fig:t-xy_offK_001} shows a much more prominent combination of fast and slow dynamics.

\section{Splitting Integrators for Massive Point Vortex Dynamics}
We would like to solve \eqref{eq:Hamilton} numerically, with a particular focus on long-time (near-)preservation of both the Hamiltonian~\eqref{eq:H} and the angular momentum~\eqref{eq:ell}.
As mentioned above, there is no explicit symplectic integrator for general non-separable Hamiltonian systems, although there are such integrators for specific classes of non-separable Hamiltonian systems~\cite{StLa2000,Bl2002,WuFoRo2003,McQu2004,Ch2009,Ta2016a,WaSuLiWu2021a,WaSuLiWu2021b,WaSuLiWu2021c,WuWaSuLi2021,WuWaSuLiHa2022}.
There are explicit integrators for an extended Hamiltonian system defined by doubling the dimension of a general non-separable one~\cite{Pi2015,Ta2016b}.
However, they are symplectic only in the extended phase space, and need to be corrected with an implicit projection to be rendered symplectic in the original phase space~\cite{JaOh2023,Oh2023b}.

\subsection{Splitting the Hamiltonian}
Our integrators are based on the following splitting of the Hamiltonian:
\begin{equation*}
 H(\br,\bp) = \frac{1}{\eps}H_{A}(\br,\bp) + H_{B}(\br,\bp)
\end{equation*}
with
\begin{subequations}
  \begin{align}
    \label{eq:H_A}
    H_{A}(\br,\bp) &\defeq \frac{1}{2} \sum_{j=1}^{N} ( \bp_{j} - q_{j} J \br_{j} )^{2},
    \\
    \label{eq:H_B}
    H_{B}(\br,\bp) &\defeq E(\br).
  \end{align}
\end{subequations}

The Hamiltonian system corresponding to $\frac{1}{\eps}H_{A}$ is then the linear system
\begin{subequations}
  \begin{equation}
    \label{eq:HamiltonA}
    \dot{\bz} = \frac{1}{\eps}\, \mathbb{J}_{2N} \nabla H_{A}(\bz)
    \iff
    \begin{bmatrix}
      \dot{\br}_{j} \smallskip\\
      \dot{\bp}_{j}
    \end{bmatrix}
    = \frac{1}{\eps}
    \begin{bmatrix}
      -q_{j} J & I \smallskip\\
      -I & -q_{j} J
    \end{bmatrix}
    \begin{bmatrix}
      \br_{j} \smallskip\\
      \bp_{j}
    \end{bmatrix},
  \end{equation}
  whereas the one with $H_{B}$ is
  \begin{equation}
    \label{eq:HamiltonB}
    \dot{\bz} = \mathbb{J}_{2N} \nabla H_{B}(\bz)
    \iff
    \left\{
      \begin{array}{l}
        \dot{\br}_{j} = 0, \smallskip\\
        \dot{\bp}_{j} = -\nabla_{j}E(\br).
      \end{array}
    \right.
  \end{equation}
\end{subequations}

\subsection{Exact Solutions of Split Systems}
Our splitting scheme to be described below is particularly simple because both systems \eqref{eq:HamiltonA} and \eqref{eq:HamiltonB} are exactly solvable.

Let us first solve \eqref{eq:HamiltonA}.
First notice that one can write the matrix on the right-hand side as the sum of two commuting matrices
\begin{equation*}
  \begin{bmatrix}
    -q_{j} J & I \\
    -I & -q_{j} J
  \end{bmatrix} =
  \begin{bmatrix}
    -q_{j} J & 0 \\
    0 & -q_{j} J
  \end{bmatrix}
  +
  \begin{bmatrix}
    0 & I \\
    -I & 0
  \end{bmatrix}.
\end{equation*}
Thus we have
\begin{align*}
  &\exp\parentheses{
  \frac{t}{\eps}
  \begin{bmatrix}
    -q_{j} J & I \\
    -I & -q_{j} J
  \end{bmatrix}
  } \\
  &= \exp\parentheses{
    \frac{t}{\eps}
    \begin{bmatrix}
      -q_{j} J & 0 \\
      0 & -q_{j} J
    \end{bmatrix}
    }
    \exp\parentheses{
    \frac{t}{\eps}
    \begin{bmatrix}
      0 & I \\
      -I & 0
    \end{bmatrix}
    } \\
             &=
               \begin{bmatrix}
                 R(q_{j}t/\eps) & 0 \\
                 0 & R(q_{j}t/\eps)
               \end{bmatrix}
               \begin{bmatrix}
                 \cos(t/\eps) I & \sin(t/\eps) I \\
                 -\sin(t/\eps) I & \cos(t/\eps) I
               \end{bmatrix} \\
             &= \begin{bmatrix}
                 \cos(t/\eps)\, R(q_{j}t/\eps)  & \sin(t/\eps)\, R(q_{j}t/\eps) \smallskip\\
                 -\sin(t/\eps)\, R(q_{j}t/\eps) & \cos(t/\eps)\, R(q_{j}t/\eps)
               \end{bmatrix},
\end{align*}
where $R(\theta) \defeq
  \begin{tbmatrix}
    \cos \theta & -\sin\theta \\
    \sin \theta & \cos\theta
  \end{tbmatrix}$.
Therefore, we may write the flow $\Phi^{A}_{t}$ of the Hamiltonian system~\eqref{eq:HamiltonA} as follows:
\begin{subequations}
  \begin{equation}
    \label{eq:PhiA}
    \Phi^{A}_{t}(\br,\bp)_{j} =
    \begin{bmatrix}
      \cos\parentheses{\frac{t}{\eps}} R\parentheses{ \frac{q_{j}t}{\eps} } \br_{j} + \sin\parentheses{\frac{t}{\eps}} R\parentheses{ \frac{q_{j}t}{\eps} } \bp_{j} \smallskip\\
      -\sin\parentheses{\frac{t}{\eps}} R\parentheses{ \frac{q_{j}t}{\eps} } \br_{j} + \cos\parentheses{\frac{t}{\eps}} R\parentheses{ \frac{q_{j}t}{\eps} } \bp_{j}
    \end{bmatrix},
  \end{equation}
  where we wrote only the $j$-th component ($1 \le j \le N$) for brevity.

  On the other hand, one easily obtains the flow $\Phi^{B}_{t}$ of the Hamiltonian system~\eqref{eq:HamiltonB} as follows:
  \begin{equation}
    \label{eq:PhiB}
    \Phi^{B}_{t}(\br,\bp)_{j} =
    \begin{bmatrix}
      \br_{j} \smallskip\\
      \bp_{j} - t\,\nabla_{j}E(\br)
    \end{bmatrix},
  \end{equation}
\end{subequations}
again showing only the $j$-th component.

Notice that the $A$-flow $\Phi^{A}_{t}$ exhibits oscillations with period $2\pi\eps$, whereas the characteristic time scale of the $B$-flow $\Phi^{B}_{t}$ is determined by $\nabla_{j} E(\br)$; it is $O(1)$ as long as the vortices do not get too close to each other.
Therefore, our splitting can be interpreted as a splitting of the dynamics of the system~\eqref{eq:Hamilton} into the fast oscillatory dynamics of the $A$-flow and the slow dynamics of the $B$-flow.

\subsection{Symplectic Integrators}
Our base method is the 2nd-order explicit integrator by the Strang splitting~\cite{St1968}:
\begin{equation}
  \label{eq:Strang}
  \Phi^{(2)}_{\dt} \defeq \Phi^{A}_{\dt/2} \circ \Phi^{B}_{\dt} \circ \Phi^{A}_{\dt/2}
\end{equation}
with time step $\dt$.
We shall refer to this method as \textsf{Split2}.

The following fundamental properties of $\Phi^{(2)}$ then follow easily from the definition:
\begin{proposition}
  \label[proposition]{prop:fundamental_properties}
  The 2nd-order integrator $\Phi^{(2)}$ defined in \eqref{eq:Strang} is symplectic and preserves the total angular momentum $\ell$ (see \eqref{eq:ell}) exactly.
\end{proposition}
\begin{proof}
  The symplecticity is clear because both $\Phi^{A}$ and $\Phi^{B}$ from \eqref{eq:PhiA} and \eqref{eq:PhiB} define Hamiltonian flows with Hamiltonians $H^{A}$ and $H^{B}$ from \eqref{eq:H_A} and \eqref{eq:H_B}, respectively.
  We also see that $\Phi^{(2)}$ preserves $\ell$ because both $\Phi^{A}$ and $\Phi^{B}$ preserve $\ell$: Notice that both $H^{A}$ and $H^{B}$ possess the $\SO(2)$-symmetry as in \eqref{eq:SO2-sym}; hence $\ell$ is a Noether invariant of both $\Phi^{A}$ and $\Phi^{B}$.
\end{proof}

We can construct higher-order integrators from \eqref{eq:Strang} using the symmetric Triple Jump composition (see \cite{CrGo1989,Forest1989,Su1990,Yo1990} and \cite[Example~II.4.2]{HaLuWa2006}):
Using the 2nd-order method in \eqref{eq:Strang}, we recursively construct an $n$th-order ($n$ being even) method as follows:
\begin{equation}
  \label{eq:nth-order-tripleJump}
  \Phi^{(n)}_{\dt} \defeq \Phi^{(n-2)}_{\gamma_{3}\dt} \circ \Phi^{(n-2)}_{\gamma_{2}\dt} \circ \Phi^{(n-2)}_{\gamma_{1}\dt},
\end{equation}
where
\begin{equation*}
  \gamma_{1} = \gamma_{3} \defeq \frac{1}{2 - 2^{1/(n-1)}},
  \qquad
  \gamma_{2} \defeq -\frac{2^{1/(n-1)}}{2 - 2^{1/(n-1)}}.
\end{equation*}
We shall refer to the 4th-order method $\Phi^{(4)}$ defined above as \textsf{Split4}.

However, for a 6th-order integrator, it is more efficient to use Yoshida's method:
\begin{equation}
  \label{eq:Yoshida}
  \Phi^{(6),\text{Y}}_{\dt} \defeq \Phi^{(2)}_{\gamma_{7}\dt} \circ \dots \circ \Phi^{(2)}_{\gamma_{2}\dt} \circ \Phi^{(2)}_{\gamma_{1}\dt},
\end{equation}
with certain values of $\gamma_{i}$'s~\cite{Yo1990} (see also \cite[Section~V.3.2]{HaLuWa2006}).
We shall refer to this method as \textsf{Split6Y}.

Since all these integrators are compositions of $\Phi^{(2)}$, it follows easily from \Cref{prop:fundamental_properties} that the above higher-order integrators share the same properties as $\Phi^{(2)}$:
\begin{corollary}
  \label[corollary]{cor:fundamental_properties}
  For every positive even integer $n$, the integrator $\Phi^{(n)}$ defined recursively by \eqref{eq:Strang} and \eqref{eq:nth-order-tripleJump} are symplectic and preserve the angular momentum $\ell$ exactly; so does $\Phi^{(6),\text{Y}}$ from \eqref{eq:Yoshida}.
\end{corollary}

\begin{remark}
  One could construct the same splitting methods for a class of systems for which the Hamiltonian is written as $H(\br,\bp) = H_{A}(\br,\bp) + H_{B} (\br)$, where $H_{A}$ is quadratic and the exponential map of the $A$-flow can be found explicitly as in \eqref{eq:PhiA}.
\end{remark}

\section{Modified Hamiltonian and Error Estimates of Hamiltonian}
\label{sec:modified_Hamiltonian}
\subsection{Modified Hamiltonian}
The above symplectic integrators do not preserve the Hamiltonian~\eqref{eq:H} exactly.
However, one can use the backward error analysis to prove that the symplectic integrators do not exhibit drifts in the Hamiltonian; this in turn implies that a $p$-th order symplectic method maintains errors in the Hamiltonian in the order of $(\dt)^{p}$ for a long time; see, e.g., \cite[Chapter~5]{LeRe2004} and \cite[Chapter~IX]{HaLuWa2006}.
This is in contrast to many other non-symplectic methods that often exhibit drifts in the Hamiltonian that result in significant errors in the Hamiltonian in the long run.

The central idea of the backward error analysis of symplectic integrators for Hamiltonian systems is to show that there is a modified Hamiltonian system
\begin{equation}
  \label{eq:Hamilton-modified}
  \dot{\bz} = \mathbb{J}_{2N}\, \nabla_{\bz} \tilde{H}(\bz; \dt),
\end{equation}
satisfied \textit{exactly} by, e.g., the flow $\Phi^{(2)}$ from \eqref{eq:Strang}, that is,
\begin{equation*}
  \od{}{t} \Phi^{(2)}_{t}(\bz) = \mathbb{J}_{2N}\, \nabla_{\bz} \tilde{H}\parentheses{ \Phi^{(2)}_{t}(\bz); \dt }.
\end{equation*}
Note that the modified Hamiltonian $\tilde{H}$ depends on the time step $\dt$.

One may prove that such $\tilde{H}$ exists for the splitting methods like ours (see, e.g., \cite[Section~5.4]{LeRe2004}).
In practice, one obtains its expressions as an asymptotic series in $\dt$; see, e.g., \cite{Ta1994}, \cite{BeGi1994}, \cite{Re1999}, \cite[Chapter~5]{LeRe2004}, \cite[Chapter~IX]{HaLuWa2006}, and references therein.

In order to estimate the order of accuracy of preservation of the Hamiltonian, one usually obtains an asymptotic expansion of the Hamiltonian in $\dt$.
However, in our setting, we are particularly interested in the regime with $\eps \ll 1$; one then usually needs to take $\dt$ smaller than $\eps$ to capture the highly oscillatory solution in the timescale of $\eps$.
Hence we would like to see how $\eps$ affects the modified Hamiltonian:

\begin{proposition}
  \label[proposition]{prop:H_error_estimate}
  For the second-order method \textsf{Split2} defined in \eqref{eq:Strang}, the modified Hamiltonian $\tilde{H}$ from \eqref{eq:Hamilton-modified} is related to the original Hamiltonian $H$ from \eqref{eq:H} as follows:
  \begin{equation}
    \label{eq:H_error_estimate}
    \tilde{H}(\bz) - H(\bz) =
    \begin{cases}
      O\parentheses{ (\dt)^{2}/\eps } & \bz \in \mathcal{K}, \smallskip\\
      O\parentheses{ (\dt)^{2}/\eps^{2} } & \bz \notin \mathcal{K}.
    \end{cases}
  \end{equation}
\end{proposition}
\begin{proof}
  Given that it is a splitting method for a Hamiltonian system, one may follow, for example, \cite[Section~5.4]{LeRe2004}) to obtain the first few terms of the asymptotic expansion of the modified Hamiltonian $\tilde{H}$:
  Using the Poisson bracket defined as
  \begin{equation*}
    \PB{F}{G} \defeq \sum_{j=1}^{N} \parentheses{ \pd{F}{\br_{j}} \cdot \pd{G}{\bp_{j}} - \pd{G}{\br_{j}} \cdot \pd{F}{\bp_{j}}},
  \end{equation*}
  we have
  \begin{align*}
    \tilde{H} &= \frac{1}{\eps} H_{A} + H_{B} - \frac{(\dt)^{2}}{24} \PB{\frac{1}{\eps}H_{A}}{\PB{\frac{1}{\eps}H_{A}}{H_{B}}} \\
              &\quad + \frac{(\dt)^{2}}{12} \PB{H_{B}}{\PB{H_{B}}{\frac{1}{\eps}H_{A}}} + O\parentheses{ (\dt)^{3} }.
  \end{align*}
  So we observe that the $O((\dt)^{2})$ leading error terms differ in scales:
  \begin{align}
    \tilde{H} - H &= - \frac{(\dt)^{2}}{24 \eps^{2}} \PB{H_{A}}{\PB{H_{A}}{H_{B}}} \nonumber\\
                  &\quad + \frac{(\dt)^{2}}{12 \eps} \PB{H_{B}}{\PB{H_{B}}{H_{A}}} + O\parentheses{ (\dt)^{3} }.
                    \label{eq:modified_H}
  \end{align}
  Specifically, the first term on the right-hand side is the leading term for the difference $\tilde{H} - H$ between the modified and the real Hamiltonians.
  
  Moreover, using the expressions~\eqref{eq:H_A} and \eqref{eq:H_B}, one finds
  \begin{align*}
    \PB{H_{A}}{\PB{H_{A}}{H_{B}}}
    = \sum_{j=1}^{N} \sum_{k=1}^{N} \bm{P}_{j}^{T} D^{2}_{jk}E(\br) \bm{P}_{k}
     + 2 \sum_{j=1}^{N} q_{j} \bm{P}_{j} \cdot (\nabla_{j}E(\br) \times \ez),
  \end{align*}
  where we defined, for each $j \in \{1, \dots, N\}$,
  \begin{equation*}
    \bm{P}_{j} \defeq \bp_{j} - q_{j} J \br_{j} \in \R^{2},
  \end{equation*}
  and $D^{2}_{jk}E$ is the $2 \times 2$ matrix defined for each pair $(j,k)$ as
  \begin{equation*}
    D^{2}_{jk}E \defeq \pd{^{2}E}{\br_{j} \partial \br_{k}}.
  \end{equation*}
  On the other hand,
  \begin{equation*}
    \PB{H_{B}}{\PB{H_{B}}{H_{A}}} = \norm{ \nabla E(\br) }^{2},
  \end{equation*}
  where $\norm{\,\cdot\,}$ is the Euclidean norm.
  Notice that $\PB{H_{A}}{\PB{H_{A}}{H_{B}}} = 0$ when $\bz = (\br,\bp) \in \mathcal{K}$ because then $\bm{P}_{j} = \bm{0}$ for $1 \le j \le N$.
  This shows that, if $\bz \in \mathcal{K}$, then the leading error term proportional to $(\dt)^{2}/\eps^{2}$ on the right-hand side of \eqref{eq:modified_H} does not contribute to the difference between $\tilde{H}$ and $H$.
\end{proof}

\subsection{Error Estimates of Hamiltonian}
The above proposition alone does not suggest that the Hamiltonian $H$ is preserved with $O((\dt)^{2}/\eps)$ if the initial condition satisfies $\bz(0) \in \mathcal{K}$, because $\bz(t) \notin \mathcal{K}$ in general, i.e., $\mathcal{K}$ is not an invariant manifold of the dynamics.
However, it is proved in \cite{OhRiGo-small_mass1} that if $\bz(0) \in \mathcal{K}$, then the deviation of the solution $\bz(t)$ from $\mathcal{K}$ is $O(\eps)$ assuming that the solution experiences no inter-vortex nor vortex--wall collisions.

More specifically, let $U \subset \mathcal{D}^{N}$ be an open subset excluding collision points
\begin{equation*}
  \Set{ (\br_{1}, \dots, \br_{N}) \in \mathcal{D}^{N} | \br_{i} = \br_{j} \text{ with some } i \neq j },
\end{equation*}
and $\mathcal{C} \subset U$ be a compact set and define
\begin{equation}
  \label{eq:c1_c2}
  c_{1} \defeq \sup_{\br \in \mathcal{C}} \norm{ \nabla E(\br) },
  \qquad
  c_{2} \defeq \sup_{\br \in \mathcal{C}} \norm{ D^{2}E(\br) },
\end{equation}
where $\norm{\,\cdot\,}$ for the Hessian is the operator norm associated with the Euclidean norm $\norm{\,\cdot\,}$ on $\R^{2N}$.
Note that both $c_{1}$ and $c_{2}$ are bounded because $E$ is a smooth function on $U$.

Then, let us define
\begin{equation}
  \label{eq:P}
  \bm{P}(t) \defeq
  \begin{bmatrix}
    \bm{P}_{1}(t) \\
    \vdots \\
    \bm{P}_{N}(t)
  \end{bmatrix}
  \in \R^{2N},
  \quad\text{where}\quad
  \bm{P}_{j}(t) \defeq \bp_{j}(t) - q_{j} J \br_{j}(t) \in \R^{2}.
\end{equation}
Notice that $\norm{ \bm{P}(t) } = 0$ if and only if $\bz(t) \in \mathcal{K}$.
So $\norm{ \bm{P}(t) }$ gives an estimate of the deviation of the solution $\bz(t)$ from $\mathcal{K}$.
We can prove that (see \cite{OhRiGo-small_mass1}) then
\begin{equation*}
  \norm{ \bm{P}(t) } \le  \eps\, c_{1}\, e^{c_{2}\,t/2},
\end{equation*}
assuming that $\bz(0) \in \mathcal{K} \cap \mathcal{C}$ and that the solution $\bz(t)$ is contained in $\mathcal{C}$; we note that $c_{1}$ and $c_{2}$ are independent of $\eps$ and $t$ as seen in their definitions in \eqref{eq:c1_c2}.

As a numerical example, \Cref{fig:t-normP_inK} shows the time evolution of the deviation $\norm{ \bm{P}(t) }$ for the solution obtained by \textsf{Split2} for the vortex dipole case considered in \Cref{ssec:dipole} below with $\eps = 0.01$ and $\bz(0) \in \mathcal{K}$.
One indeed sees that $\norm{\bm{P}(t)} = O(\eps)$.
\begin{figure}[htbp]
  \centering
  \includegraphics[width=.8\textwidth]{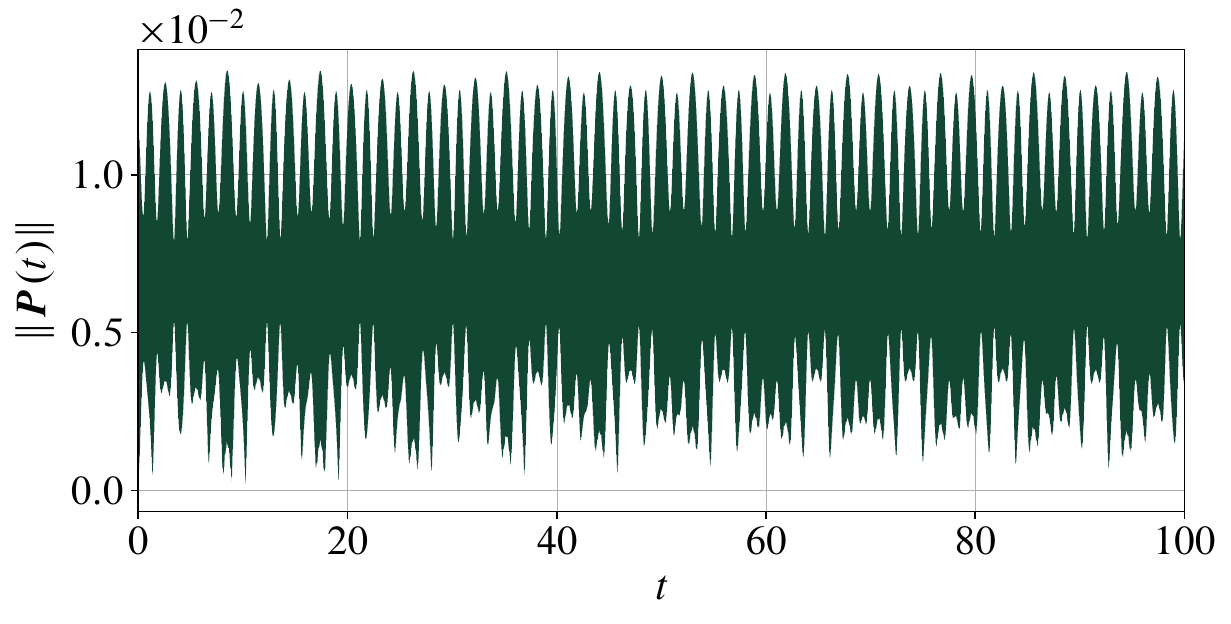}
  \caption{
    Time evolution of deviation $\norm{\bm{P}(t)}$ of massive solution $\bz(t)$ obtained by \textsf{Split2} with parameters and initial conditions from \eqref{eq:dipole_IC_in_K} below (vortex dipole case) with $\eps = 0.01$ and $\bz(0) \in \mathcal{K}$.
    One can see that $\norm{\bm{P}(t)} = O(\eps)$ for a long time.
  }
  \label{fig:t-normP_inK}
\end{figure}

Assuming that the numerical solution indeed satisfies this property, we have the following error estimate of the Hamiltonian:
\begin{theorem}
  Suppose that the numerical solution $\{ \bz_{n} \}$ of \textsf{Split2} defined in \eqref{eq:Strang} with initial condition $\bz_{0} \in \mathcal{K}$ stays in a compact subset $\mathcal{C} \subset U$ described above, and also that the corresponding $\{ \bm{P}_{n} \defeq \bm{P}(\bz_{n}) \}$ (see \eqref{eq:P}) satisfies $\norm{\bm{P}_{n}} = O(\eps)$.
  Then we have
  \begin{equation*}
    H(\bz_{n}) - H(\bz_{0}) = O\parentheses{ (\dt)^{2}/\eps }.
  \end{equation*}
\end{theorem}
\begin{proof}
  Following the standard approach in proving the long-time preservation of the Hamiltonian of \cite{BeGi1994} (see also \cite[Section~5.2]{LeRe2004} and \cite[Section~IX.8]{HaLuWa2006}), one may find an optimal truncation $\overline{H}$ of the asymptotic expansion for $\tilde{H}$ as follows:
  \begin{equation*}
    \overline{H} \defeq H + (\dt)^{2} \tilde{H}^{(2)} + \dots + (\dt)^{\mathcal{N}} \tilde{H}^{(\mathcal{N})},
  \end{equation*}
  so that the flow $\overline{\Phi}$ of the corresponding modified Hamiltonian system
  \begin{equation*}
    \dot{\bz} = \mathbb{J}_{2N}\, \nabla_{\bz} \overline{H}(\bz; \dt),
  \end{equation*}
  satisfies
  \begin{equation*}
    \norm{ \Phi^{(2)}_{\dt}(\bz) - \overline{\Phi}_{\dt} } \le \frac{c_{3}}{\eps}\,\dt\,e^{-c_{4}\,\eps/\dt}
  \end{equation*}
  with some constants $c_{3}, c_{4} > 0$ that are independent of $\eps$ but are $O(N)$ due to the inter-vortex interaction terms; each vortex interacts with $N - 1$ others; we note that the two occurrences of $\eps$ come from the $1/\eps$ terms of the vector field of the system~\eqref{eq:Hamilton}; see, e.g., \cite[Theorem~IX.7.6]{HaLuWa2006} on the details of estimates of this type.
  Then one can show that
  \begin{equation*}
    \abs{ \overline{H}(\bz_{n}) - \overline{H}(\bz_{0}) } \le \frac{\lambda\,c_{3}}{\eps^{2}}\, n \dt\, e^{-c_{4}\,\eps/\dt},
  \end{equation*}
  where $\lambda/\eps > 0$ is a Lipschitz constant for $\overline{H}$, where the $\eps$-dependence comes from the $1/\eps$ terms of the Hamiltonian~\eqref{eq:H}.
  Therefore,
  \begin{align*}
    H(\bz_{n}) - H(\bz_{0})
    &= (\dt)^{2} \parentheses{ \tilde{H}^{(2)}(\bz_{n}) - \tilde{H}^{(2)}(\bz_{0}) } \\
    &\quad + (\dt)^{3} (F(\bz_{n}) - F(\bz_{0})) + O\parentheses{ n (\dt/\eps^{2})\, e^{-c_{4}\,\eps/\dt} },
  \end{align*}
  where, as found in \eqref{eq:modified_H},
  \begin{equation*}
    \tilde{H}^{(2)} \defeq -\frac{1}{24 \eps^{2}} \PB{H_{A}}{\PB{H_{A}}{H_{B}}}
    + \frac{1}{12 \eps} \PB{H_{B}}{\PB{H_{B}}{H_{A}}},
  \end{equation*}
  and
  \begin{equation*}
    F \defeq  \tilde{H}^{(3)} + \dots + (\dt)^{\mathcal{N}-3} \tilde{H}^{(\mathcal{N})},
  \end{equation*}
  and each $\tilde{H}^{(k)}$ is a linear combination of $k$-time repeated Poisson brackets of $\frac{1}{\eps} H_{A}$ and $H_{B}$.
  Since both $H_{A}$ and $H_{B}$ are smooth in $U$, so are all repeated Poisson brackets of them; hence they are are bounded on $\mathcal{C}$, and hence so is $F$.
  
  However,
  \begin{align*}
    \tilde{H}^{(2)}(\bz_{n}) - \tilde{H}^{(2)}(\bz_{0})
    &= -\frac{1}{24 \eps^{2}} \PB{H_{A}}{\PB{H_{A}}{H_{B}}}(\bz_{n}) \\
    &\quad + \frac{1}{12 \eps} \parentheses{ \PB{H_{B}}{\PB{H_{B}}{H_{A}}}(\bz_{n}) - \PB{H_{B}}{\PB{H_{B}}{H_{A}}}(\bz_{0}) },
  \end{align*}
  because $\bz_{0} \in \mathcal{K}$ and $\PB{H_{A}}{\PB{H_{A}}{H_{B}}}$ vanishes on $\mathcal{K}$ as observed in the proof of \Cref{prop:H_error_estimate}.
  However, since $\norm{\bm{P}_{n}} = O(\eps)$ implies that $\bz_{n}$ is $O(\eps)$ away from $\mathcal{K}$, all the functions are smooth in $U$, and $\bz_{0}, \bz_{n}$ are in compact set $\mathcal{C}$, we have
  \begin{equation*}
    \PB{H_{A}}{\PB{H_{A}}{H_{B}}}(\bz_{n}) = O(\eps)
  \end{equation*}
  and that both $\PB{H_{B}}{\PB{H_{B}}{H_{A}}}(\bz_{n})$ and $\PB{H_{B}}{\PB{H_{B}}{H_{A}}}(\bz_{0})$ are $O(1)$.
  Hence we conclude that $\tilde{H}^{(2)}(\bz_{n}) - \tilde{H}^{(2)}(\bz_{0}) = O(1/\eps)$.
\end{proof}

We shall illustrate the above theorem numerically in the examples to follow.

\section{Single Massive Vortex ($N = 1$)}
\subsection{Testing with Single Massive Vortex}
Let us test the integrators using the single vortex example discussed in \Cref{ssec:single_vortex}.
Using the polar coordinates $(r, \theta)$ for $\br$ and $(p_{r}, p_{\theta})$ for $\bp$, the Hamiltonian~\eqref{eq:H} becomes
\begin{equation*}
  H = \frac{1}{2\eps}\parentheses{ p_{r}^{2} + \parentheses{ \frac{\ell}{r} + q_{1} r }^{2} } + \ln\parentheses{ 1 - r^{2} },
\end{equation*}
where $\ell$, the angular momentum, is an invariant of the system.
Since the above expression of $H$ depends only on $(r, p_{r})$, the level set of $H$ at its initial value on $(r,p_{r})$-plane gives the trajectory $(r(t),p_{r}(t))$.

As in \Cref{ssec:single_vortex}, we set $q_{1} = 1$ and $\eps = 0.01$, and consider the initial conditions \eqref{eq:singlevort_IC_in_K} and \eqref{eq:singlevort_IC_off_K}, for which $\bz(0) = (\br(0), \bp(0)) \in \mathcal{K}$ and $\bz(0) \notin \mathcal{K}$, respectively.
Then the level set of $H$ gives a closed curve in each case.

\begin{figure}[htbp]
  \centering
    \begin{subcaptionblock}[b]{.52\linewidth}
    \centering
    \includegraphics[width=\linewidth]{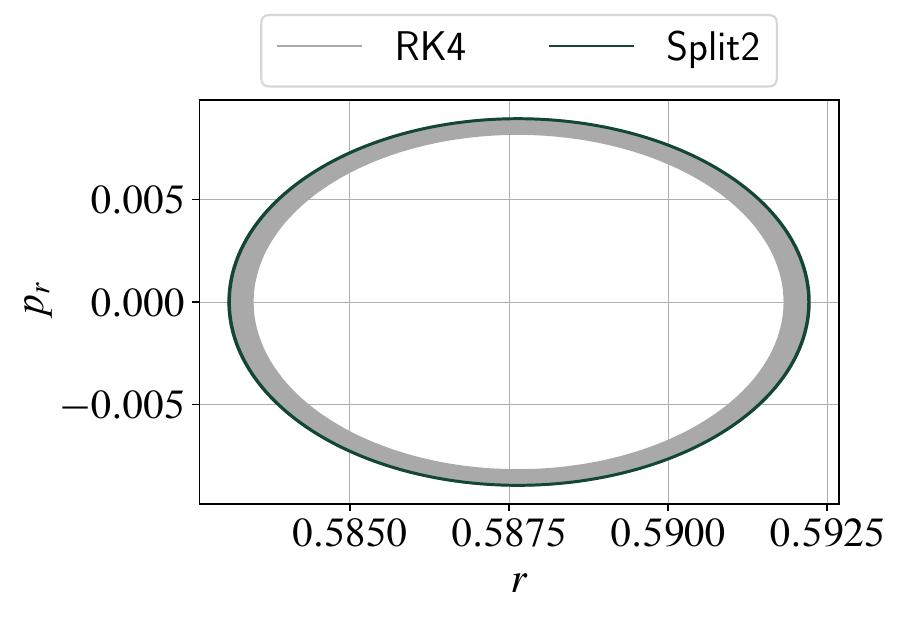}
    \caption{Initial condition~\eqref{eq:singlevort_IC_in_K} with $\bz(0) \in \mathcal{K}$.}
    \label{fig:xy_inK_001} 
  \end{subcaptionblock}
  \begin{subcaptionblock}[b]{.47\linewidth}
    \centering
    \includegraphics[width=\linewidth]{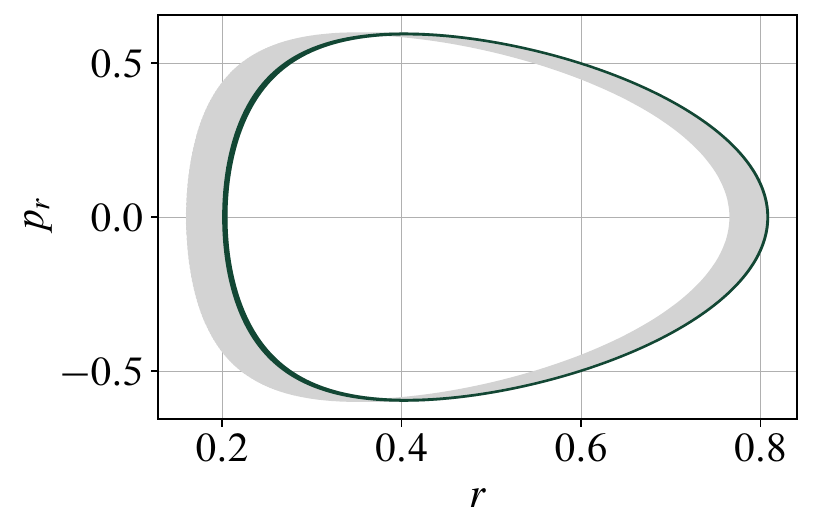}
    \caption{Initial condition~\eqref{eq:singlevort_IC_off_K} with $\bz(0) \notin \mathcal{K}$.}
    \label{fig:xy_offK_001}
  \end{subcaptionblock}
  \caption{Phase portraits on the $(r,p_{r})$-plane of single massive vortex dynamics for $0 \le t \le 200$ computed by Runge--Kutta method (\textsf{RK4}) as well as 2nd-order splitting method~\eqref{eq:Strang} (\textsf{Split2}); $q_{1} = 1$, $\eps = 0.01$, and $\dt = 10^{-3}$.}
  \label{fig:phaseportrait_001}
\end{figure}

\Cref{fig:phaseportrait_001} shows the trajectories or the (projected) phase portraits $(r(t),p_{r}(t))$ for $0 \le t \le 200$ using the standard (4th-order) Runge--Kutta method (\textsf{RK4}) as well as our 2nd-order splitting method~\eqref{eq:Strang} (\textsf{Split2}) with the initial conditions \eqref{eq:singlevort_IC_in_K} and \eqref{eq:singlevort_IC_off_K}, and $\dt = 10^{-3}$.
One observes that the \textsf{RK4} solution significantly deviates from a closed curve, especially in the latter case with $\bz(0) \notin \mathcal{K}$.
On the other hand, the \textsf{Split2} solution exhibits much smaller deviation from a closed curve, despite being a lower-order method than \textsf{RK4}.
Notice also the difference in scales in the two plots: The drift in the \textsf{RK4} solution in the latter case is far greater than that of the former.

\Cref{fig:t-H_001} shows the time evolution of the error in the Hamiltonian $|H(t) - H_{0}|$ where $H_{0} \defeq H(\br(0), \bp(0))$ is the initial value of Hamiltonian $H$.

\begin{figure}[htbp]
  \centering
    \begin{subcaptionblock}[b]{.49\linewidth}
    \centering
    \includegraphics[width=\linewidth]{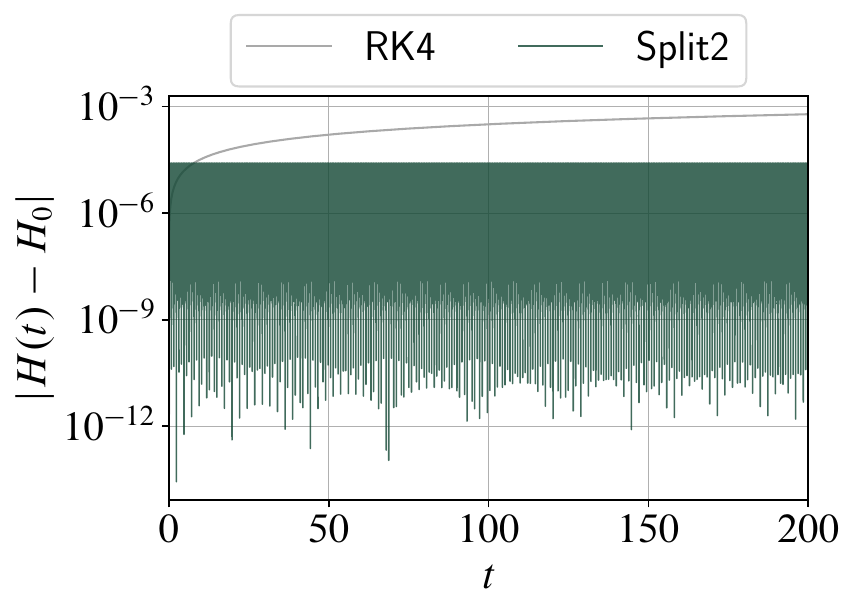}
    \caption{Initial condition~\eqref{eq:singlevort_IC_in_K} with $\bz(0) \in \mathcal{K}$.}
    \label{fig:t-H_inK_001}
  \end{subcaptionblock}
  \begin{subcaptionblock}[b]{.49\linewidth}
    \centering
    \includegraphics[width=\linewidth]{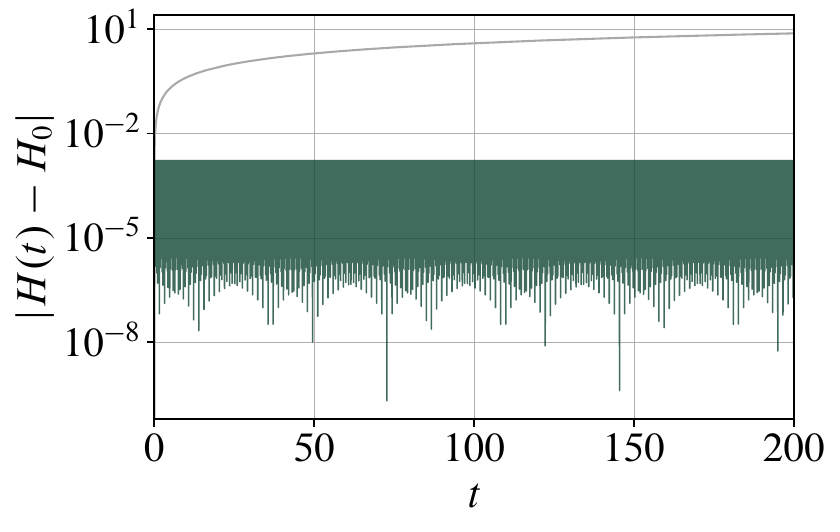}
    \caption{Initial condition~\eqref{eq:singlevort_IC_off_K} with $\bz(0) \notin \mathcal{K}$.}
    \label{fig:t-H_offK_001}
  \end{subcaptionblock}
  \caption{Time evolution of errors in Hamiltonian $H$ for the same single-vortex problem from \Cref{fig:phaseportrait_001}. The errors with \textsf{Split2} are in agreement with the prediction discussed in \Cref{sec:modified_Hamiltonian}.}
  \label{fig:t-H_001}
\end{figure}

Recall from \Cref{prop:H_error_estimate} that the leading order term in the difference between the modified Hamiltonian $\tilde{H}(\bz)$ and the original Hamiltonian $H(\bz)$ differs depending on whether $\bz$ is in $\mathcal{K}$ or not: It is proportional to $(\dt)^{2}/\eps^{2}$ if $\bz \notin \mathcal{K}$ whereas $(\dt)^{2}/\eps$ if  $\bz \in \mathcal{K}$.
Since $\eps = 10^{-2}$ and $\dt = 10^{-3}$, we have $(\dt)^{2}/\eps = 10^{-4}$ whereas $(\dt)^{2}/\eps^{2} = 10^{-2}$.
Thus we expect $|H(t) - H_{0}|$ to be in the order of $10^{-4}$ when $\bz(0) \in \mathcal{K}$ and $10^{-2}$ when $\bz(0) \notin \mathcal{K}$.
\Cref{fig:t-H_inK_001} and \Cref{fig:t-H_offK_001} indeed show that the maximum errors are in those scales.

\subsection{What if $\dt \ge \eps$?}
In the above test, we picked $\dt = \eps/10$ to resolve the highly oscillatory motion of $O(\eps)$-scale.
What if $\dt$ is not small enough to resolve such oscillations?
\Cref{fig:phaseportrait_001-2} shows the phase portraits for the same two cases with $\bz(0) \in \mathcal{K}$ and $\bz(0) \notin \mathcal{K}$, but with two time steps $\dt = 0.01 = \eps$ and $\dt = 0.05 = 5\eps$.
For $\dt = 0.01$, the solutions computed by \textsf{RK4} deviate from the energy contour very quickly and spiral inside the contour, whereas the \textsf{Split2} solutions seem more stable, albeit exhibiting slow but eventually very significant drift.
For $\dt = 0.05$, the \textsf{RK4} diverged very quickly and hence are not shown in the figures, whereas the \textsf{Split2} solutions behave in a similar way as in $\dt = 0.01$ with an even more significant drift.
Nevertheless, these numerical experiments show that \textsf{Split2} is more stable than \textsf{RK4} with $\dt \ge \eps$, and fails more gracefully as $\dt$ becomes greater than $\eps$.
\begin{figure}[htbp]
  \centering
  \begin{subcaptionblock}[b]{.49\linewidth}
    \centering
    \includegraphics[width=\linewidth]{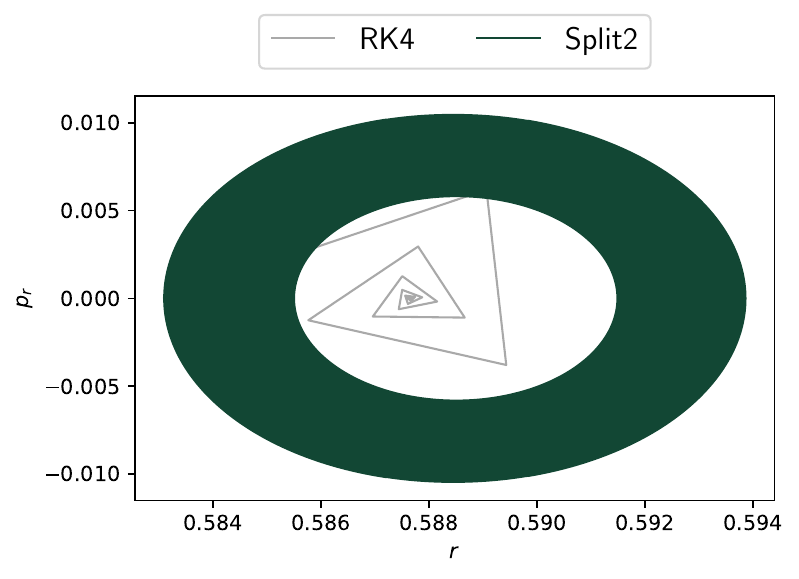}
    \caption{$\bz(0) \in \mathcal{K}$ and $\dt = 0.01$.}
    \label{fig:xy_inK_001_1e-2}
  \end{subcaptionblock}
  \begin{subcaptionblock}[b]{.49\linewidth}
    \centering
    \includegraphics[width=\linewidth]{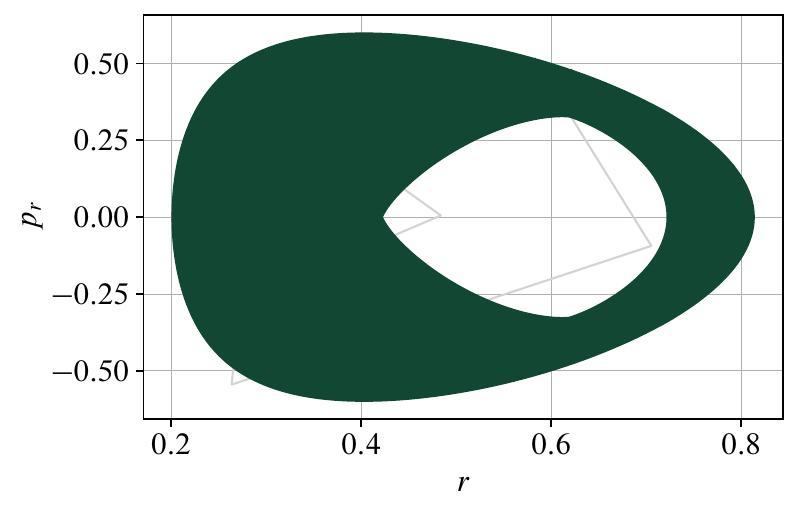}
    \caption{$\bz(0) \notin \mathcal{K}$ and $\dt = 0.01$.}
    \label{fig:xy_offK_001_1e-2}
  \end{subcaptionblock}
  \vspace{1ex}\\
  \begin{subcaptionblock}[b]{.49\linewidth}
    \centering
    \includegraphics[width=\linewidth]{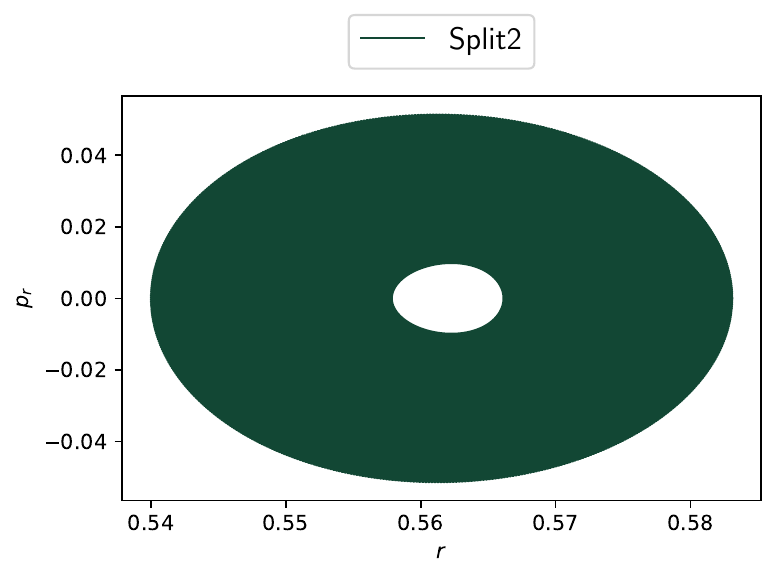}
    \caption{$\bz(0) \in \mathcal{K}$ and $\dt = 0.05$.}
    \label{fig:xy_inK_001_5e-2}
  \end{subcaptionblock}
  \begin{subcaptionblock}[b]{.49\linewidth}
    \centering
    \includegraphics[width=\linewidth]{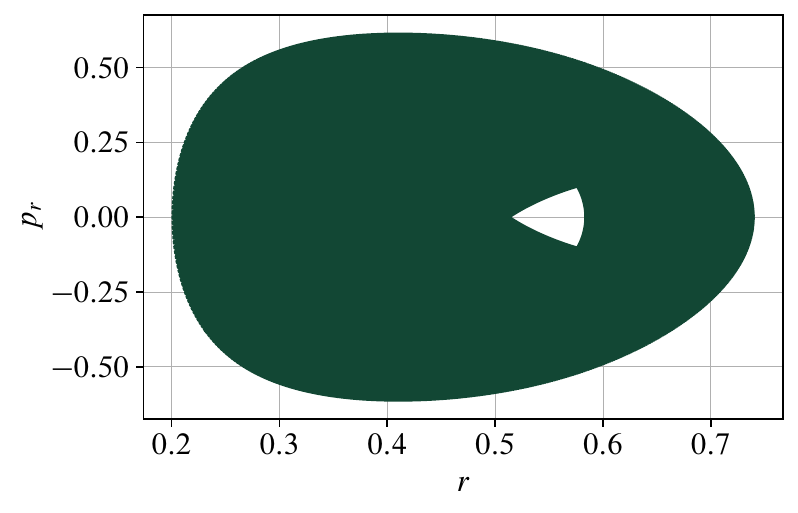}
    \caption{$\bz(0) \notin \mathcal{K}$ and $\dt = 0.05$.}
    \label{fig:xy_offK_001_5e-2}
  \end{subcaptionblock}
  \caption{Same as \Cref{fig:phaseportrait_001} (single massive point vortex) except that the time step $\dt$ is taken as $0.01$ or $0.05$ so that $\dt \ge \eps$.
    The results with \textsf{RK4} with $\dt = 0.05$ are not shown because they quickly diverge.
  }
  \label{fig:phaseportrait_001-2}
\end{figure}

\begin{remark}
  We note in passing that it is not very practical to consider numerical solutions of the massive system~\eqref{eq:Hamilton} for $\eps \ll 1$ with $\dt \ge \eps$ in general for two reasons.
  Firstly, the numerical solutions would be inaccurate as illustrated in \Cref{fig:phaseportrait_001-2} even with our splitting integrators.
  Secondly, and perhaps more importantly, taking $\dt \ge \eps$ essentially implies that one disregards the fast oscillations in the solutions, and hence is effectively interested in the solution behavior after averaging out the fast oscillations.
  However, it is proved in \cite{OhRiGo-small_mass1} that a certain form of averaging of the massive equations~\eqref{eq:Hamilton} yields the well-known (massless) point vortex equations~\eqref{eq:Kirchhoff} assuming that the initial condition is close to $\mathcal{K}$.
  More specifically, the massless solution gives an approximation of the massive solution with $O(\eps)$ accuracy.
  So if one is mainly interested in such averaged behaviors of the massive system~\eqref{eq:Hamilton} with $\eps \ll 1$, then it is much more effective to solve the averaged (massless) equations~\eqref{eq:Kirchhoff} numerically: They are a half the dimension of the massive ones~\eqref{eq:Hamilton}, and their dynamics is ``slow'' without fast oscillations; hence one does not have to worry about taking $\dt$ smaller than $\eps$.
\end{remark}

\section{Massive Vortex Dipole ($N = 2$)}
\subsection{Massive Vortex Dipole}
\label{ssec:dipole}
Consider the vortex dipole case with the following parameters and initial conditions:
\begin{equation}
  \label{eq:dipole_IC_in_K}
  \begin{array}{c}
    N = 2,\quad  q_{1} = -1,\quad q_{2} = 1,\quad \eps = 0.01, \medskip\\
    \br_{1}(0) =
    \begin{bmatrix}
      0.6 \\
      0.2
    \end{bmatrix},
    \quad
    \br_{2}(0)=
    \begin{bmatrix}
      -0.3 \\
      -0.4
    \end{bmatrix},
    \medskip\\
    \bp_{1}(0) = q_{1} J \br_{1}(0) =
    \begin{bmatrix}
      -0.2 \\
      0.6
    \end{bmatrix},
    \medskip\\
    \bp_{2}(0) = q_{2} J \br_{2}(0) =
    \begin{bmatrix}
      -0.4 \\
      0.3
    \end{bmatrix}.
  \end{array}
\end{equation}
Notice that $\bp_{j}(0) = q_{j} J \br_{j}(0)$ so that $\bz(0) \in \mathcal{K}$.

We also consider another set of initial conditions with the same conditions as above except
\begin{equation}
  \label{eq:dipole_IC_off_K}
  \begin{array}{c}
    \bp_{1}(0) = q_{1} J \br_{1}(0) +
    \begin{bmatrix}
      -0.15 \\
      0.125
    \end{bmatrix}
    =
    \begin{bmatrix}
      -0.35 \\
      0.725
    \end{bmatrix},
    \medskip\\
    \bp_{2}(0) = q_{2} J \br_{2}(0) +
    \begin{bmatrix}
      0.075 \\
      0.2
    \end{bmatrix}
    =
    \begin{bmatrix}
      -0.325 \\
      0.5
    \end{bmatrix},
  \end{array}
\end{equation}
which gives $\bz(0) \notin \mathcal{K}$.

\subsection{Comparison of Trajectories}
\Cref{fig:x-y_001} shows the trajectories of both vortices for the above two sets of initial conditions, computed by \textsf{Split6Y}.
Just as we saw in \Cref{fig:t-xy_001} for the single vortex case, the trajectories have only small fluctuations in the former case with $\bz(0) \in \mathcal{K}$.
On the other hand, for the latter case with $\bz(0) \in \mathcal{K}$, the trajectories are highly oscillatory, clearly exhibiting the separation of scales as we have observed in \Cref{fig:t-xy_001} for the single vortex case.

\begin{figure}[htbp]
  \centering
  \begin{subcaptionblock}{.5\linewidth}
    \centering
    \includegraphics[width=\linewidth]{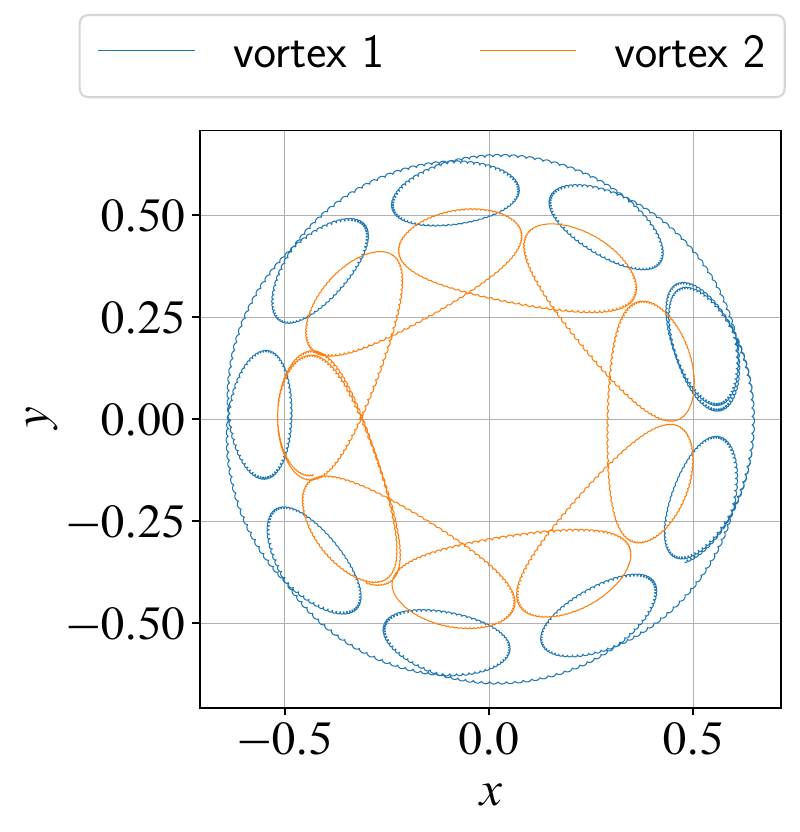}
    \caption{Initial condition~\eqref{eq:dipole_IC_in_K} with $\bz(0) \in \mathcal{K}$.}
    \label{fig:x-y_inK_001}
  \end{subcaptionblock}
  \begin{subcaptionblock}{.485\linewidth}
    \centering
    \includegraphics[width=\linewidth]{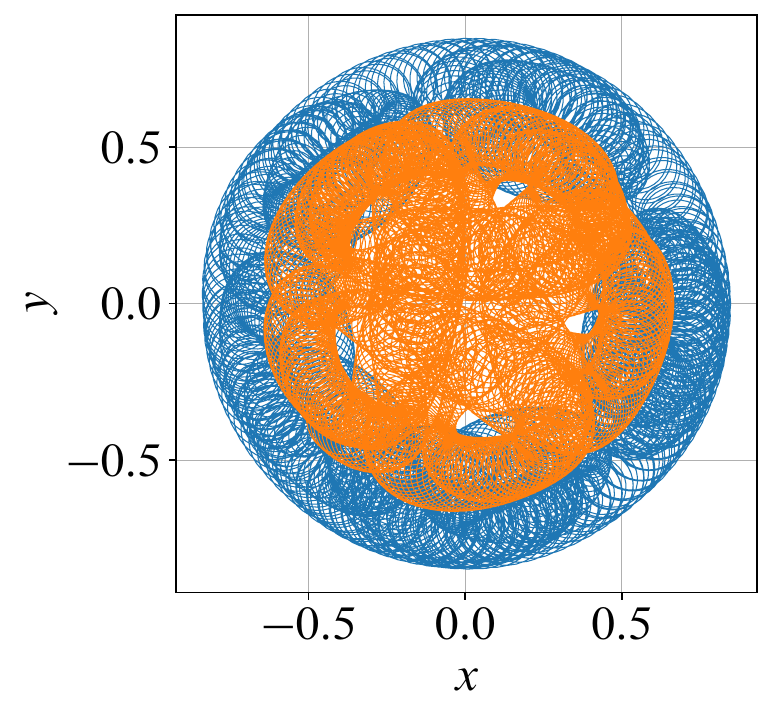}
    \caption{Initial condition~\eqref{eq:dipole_IC_off_K} with $\bz(0) \notin \mathcal{K}$.}
    \label{fig:x-y_offK_001}
  \end{subcaptionblock}
  \caption{Trajectories of massive vortex dipole; $N = 2$, $q_{1} = -1$, $q_{2} = 1$, $\eps = 0.01$, and $\dt = 10^{-3}$; computed by \textsf{Split6Y} and plotted for $0 \le t \le 30$.}
  \label{fig:x-y_001}
\end{figure}

\subsection{Comparison of Errors in Invariants}
\Cref{fig:errors_001} shows the time evolution of relative errors of two invariants---the Hamiltonian $H$ from \eqref{eq:H} and the angular momentum $\ell$ from \eqref{eq:ell}---as well as the deviation $\norm{\bm{P}(t)}$ from $\mathcal{K}$ (see \eqref{eq:P}) for $0 \le t \le 100$ with the above initial conditions, using $\textsf{RK4}$, $\textsf{Split2}$, $\textsf{Split4}$, and $\textsf{Split6Y}$.
We set the initial values of the invariants as
\begin{equation*}
  H_{0} \defeq H(\br(0), \bp(0)),
  \qquad
  \ell_{0} \defeq \ell(\br(0), \bp(0)).
\end{equation*}

\begin{figure}[htbp]
  \centering
  \begin{subcaptionblock}{\linewidth}
    \centering
    \includegraphics[width=\linewidth]{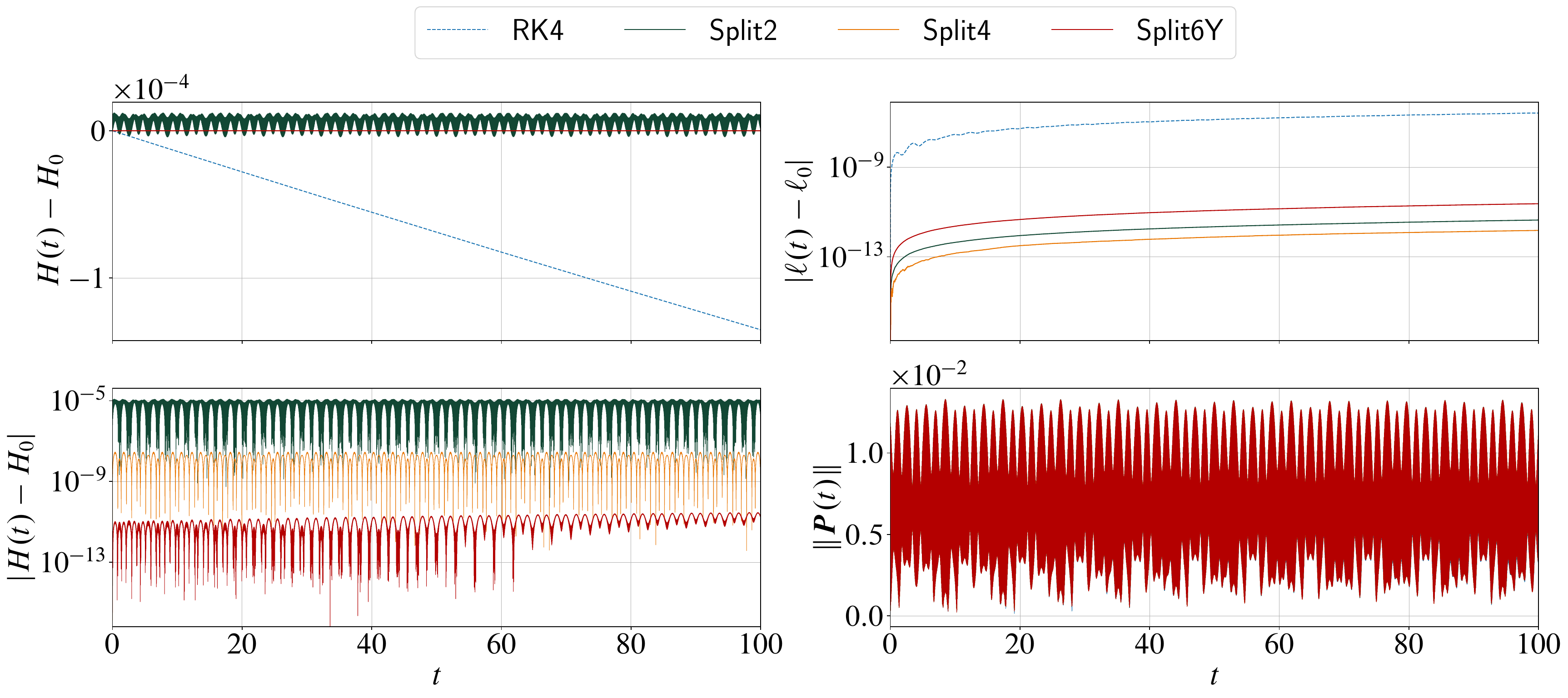}
    \caption{Initial condition~\eqref{eq:dipole_IC_in_K} with $\bz(0) \in \mathcal{K}$.}
    \label{fig:errors_inK_001}
  \end{subcaptionblock}
  \vspace{1ex}\\
  \begin{subcaptionblock}{\linewidth}
    \centering
    \includegraphics[width=\linewidth]{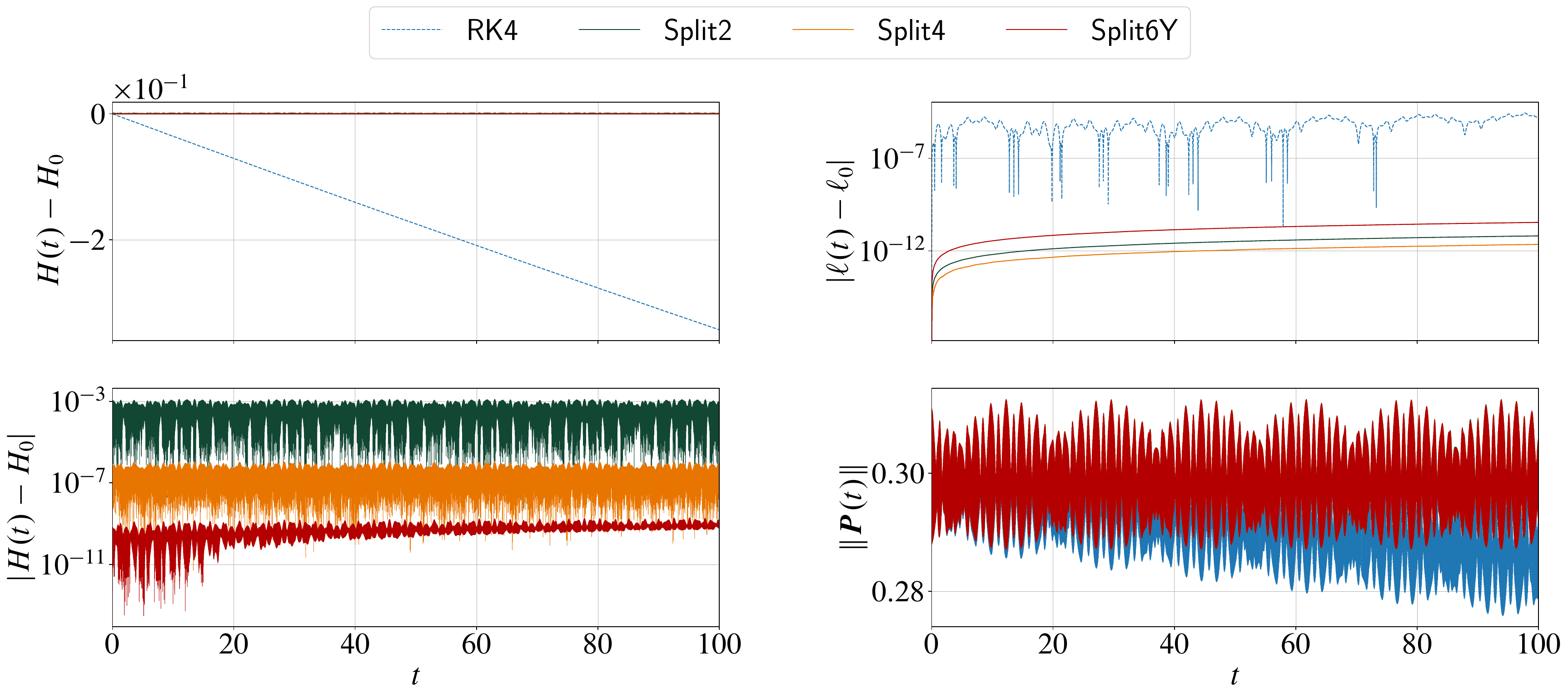}
    \caption{Initial condition~\eqref{eq:dipole_IC_off_K} with $\bz(0) \notin \mathcal{K}$.}
    \label{fig:errors_offK_001}
  \end{subcaptionblock}
  \caption{Time evolution of relative errors in Hamiltonian $H$, angular momentum $\ell$, and deviation $\norm{\bm{P}(t)}$ from $\mathcal{K}$; $N = 2$, $q_{1} = -1$, $q_{2} = 1$, $\eps = 0.01$, and $\dt = 10^{-3}$. The plots of the deviation $\norm{\bm{P}(t)}$ (not an invariant of the system) for schemes other than \textsf{Split6Y} are indistinguishable.}
  \label{fig:errors_001}
\end{figure}

For the former case with $\bz(0) \in \mathcal{K}$, one observes drifts in both $H$ and $\ell$ for the \textsf{RK4} solution.
On the other hand, the Hamiltonian for all the splitting integrators exhibit only small fluctuations near $H_{0}$ without any drifts, just as observed in \Cref{fig:phaseportrait_001}.
Recall from \Cref{prop:fundamental_properties} and \Cref{cor:fundamental_properties} that the splitting integrators preserve $\ell$ exactly.
One can see that the errors in $\ell$ for the splitting integrators are indeed negligibly small compared to that for \textsf{RK4}.

We note that the the plots of the deviation $\norm{\bm{P}(t)}$ (not an invariant of the system) for different schemes are indistinguishable because it is in the scale of $\max\{ \eps, \norm{\bm{P}(0)} \}$ ($\norm{\bm{P}(0)} = 0$ in panel \subref{fig:errors_inK_001} and $\norm{\bm{P}(0)} \simeq 0.29$ in \subref{fig:errors_offK_001}).
Compare the plot on the lower-right in the panel~\subref{fig:errors_inK_001} with \Cref{fig:t-normP_inK}: the latter is only for \textsf{Split2}, and is virtual identical to the former, in which only the \textsf{Split6Y} result is visible.

For the case with $\bz(0) \notin \mathcal{K}$, one sees that the drift in $H$ for \textsf{RK4} is significantly greater than the former case: the relative error grows to the order of $10^{-1}$ before $t = 100$ (in contrast to $10^{-4}$ in the former case).
The relative errors in $H$ for the splitting integrators have grown roughly by the multiplicative factor of $10^{2}$ in comparison to the former case.
This again confirms our prediction using the modified Hamiltonian that the error in $H$ for $\bz(0) \notin \mathcal{K}$ is greater than that with $\bz(0) \in \mathcal{K}$ by the factor of $1/\eps$, given that $\eps = 10^{-2}$ here.
However, notice that the relative errors still remain quite small compared to $10^{-1}$.
In particular, \textsf{Split4}---4th-order method just like \textsf{RK4}---maintains relative errors in the scale of $10^{-7}$.

\section{Massive Vortex Necklace ($N = 5$)}
\subsection{Pentagonal Necklace Solution}
As an additional example, we shall consider the massive vortex necklace solutions considered in \cite{PhysRevResearch.5.023109,BePe2025}, in particular the case with $N = 5$ here.
As the initial condition, we consider massive point vortices located at the vertices of the pentagon:
\begin{equation}
  \label{eq:r-Necklace}
  \br_{j}(0) = \rho ( \cos\theta_{j}, \sin\theta_{j} )
  \quad
  \text{with}
  \quad
  \theta_{j} \defeq \frac{2(j-1)}{N} \pi
  \quad
  \text{for}
  \quad
  1 \le j \le N,
\end{equation}
where $\rho \in (0,1)$ is the radius of the circle inscribed by the pentagon.
Let $\omega \in \R$, and suppose one sets the initial momenta $\bp_{j}(0)$ to be $(q_{j} - \eps\,\omega) J \br_{j}(0)$ so that, in view of \eqref{eq:p}, the initial velocities $\dot{\br}_{j}(0)$ become $-\omega J \br_{j} = \omega (\ez \times \br_{j})$.
One can then show that the vortices exhibit a rigid counterclockwise rotation with angular velocity $\omega$ about the origin assuming that $q_{j} = 1$ for $1 \le j \le N$ and also if the quadratic equation
\begin{equation}
  \label{eq:omega}
  \eps\,\rho^{2}(1 - \rho^{10}) \omega^{2} - 2\rho^{2}(1 - \rho^{10}) \omega + 6\rho^{10} + 4 = 0
\end{equation}
for $\omega$ admits real solutions; see \cite{PhysRevResearch.5.023109} for a more general formula for the $N$-vortex necklace solution and its details.

\subsection{Perturbed Pentagonal Necklace Solutions}
Since the rigid rotation is not particularly interesting for a numerical experiment, we shall add some perturbation to the initial momenta as follows:
\begin{equation}
  \label{eq:p-Necklace}
  \bp_{j}(0) = (q_{j} - \eps\,\omega) J \br_{j}(0) + \frac{1}{10} \br_{j}(0)
  \quad
  \text{for}
  \quad
  1 \le j \le N,
\end{equation}
where $\br_{j}(0)$ are given in \eqref{eq:r-Necklace}; the perturbation proportional to $\br_{j}(0)$ gives a ``kick'' to push the solutions off the circular trajectories of the rigid rotation.
We note that, for the initial condition $\bz(0) = (\br(0),\bp(0))$ with \eqref{eq:r-Necklace} and \eqref{eq:p-Necklace}, one has $\norm{\bm{P}(0)} \simeq 0.16$ and so $\bz(0) \notin \mathcal{K}$.

\Cref{fig:x-y_necklace} shows the trajectories of the vortices computed using \textsf{Split6Y} for $0 \le t \le 0.3$ with $q_{j} = 1$ for $1 \le j \le 5$, $\eps = 0.01$, $\rho = 0.5$; the angular velocity $\omega \simeq 8.37$ is obtained as one of the solutions to \eqref{eq:omega} for the given parameters.
\begin{figure}[htbp]
  \centering
  \includegraphics[width=.9\textwidth]{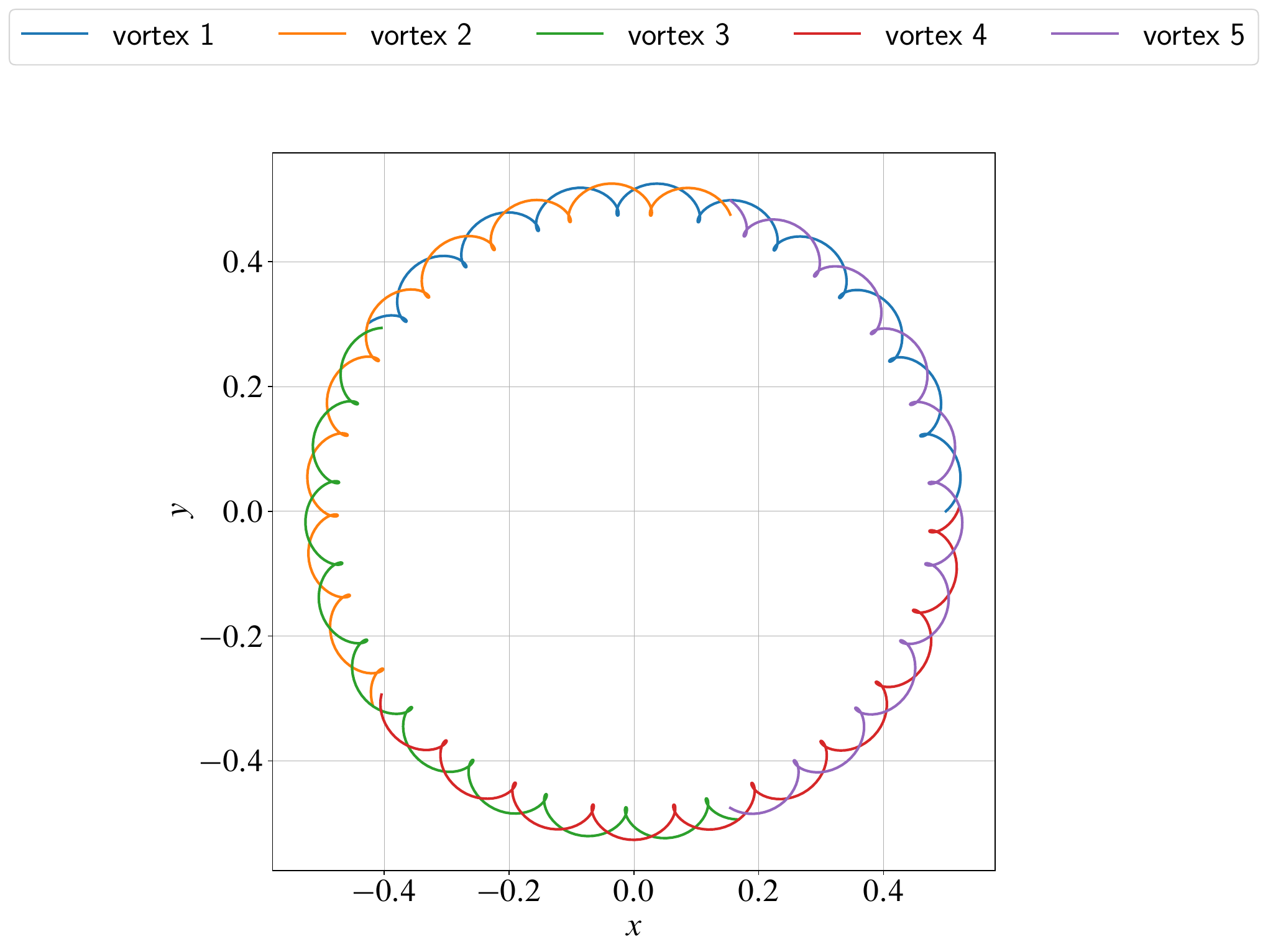}
  \caption{
    Trajectories ($0 \le t \le 0.3$) of five massive vortices with initial condition $\bz(0) = (\br(0),\bp(0))$ given by \eqref{eq:r-Necklace} and \eqref{eq:p-Necklace}; the parameters are $q_{j} = 1$ for $1 \le j \le 5$, $\eps = 0.01$, $\rho = 0.5$, and $\omega \simeq 8.37$ satisfying \eqref{eq:omega}.
    The solutions are computed with \textsf{Split6Y} with $\dt = 10^{-3}$.
    The initial positions are in the vertices of a pentagon, and the momenta are slightly perturbed from those for the necklace (rigid rotation) solutions, resulting in spiral motions 
  }
  \label{fig:x-y_necklace}
\end{figure}

\Cref{fig:errors_necklace} shows the time evolution of relative errors of two invariants---the Hamiltonian $H$ from \eqref{eq:H} and the angular momentum $\ell$ from \eqref{eq:ell}---as well as the deviation $\norm{\bm{P}(t)}$ from $\mathcal{K}$ for the pentagonal necklace problem from \Cref{fig:x-y_necklace}.
Recall that, in \Cref{fig:errors_offK_001}, we had $\bz(0) \notin \mathcal{K}$ as in this case.
One notices that the errors in the Hamiltonian are roughly in the same scale as in those \Cref{fig:errors_offK_001}, numerically demonstrating that our estimate of errors in the Hamiltonian for the case $\bz(0) \notin \mathcal{K}$ is valid in this higher-dimensional case.
\begin{figure}[htbp]
  \centering
  \includegraphics[width=\textwidth]{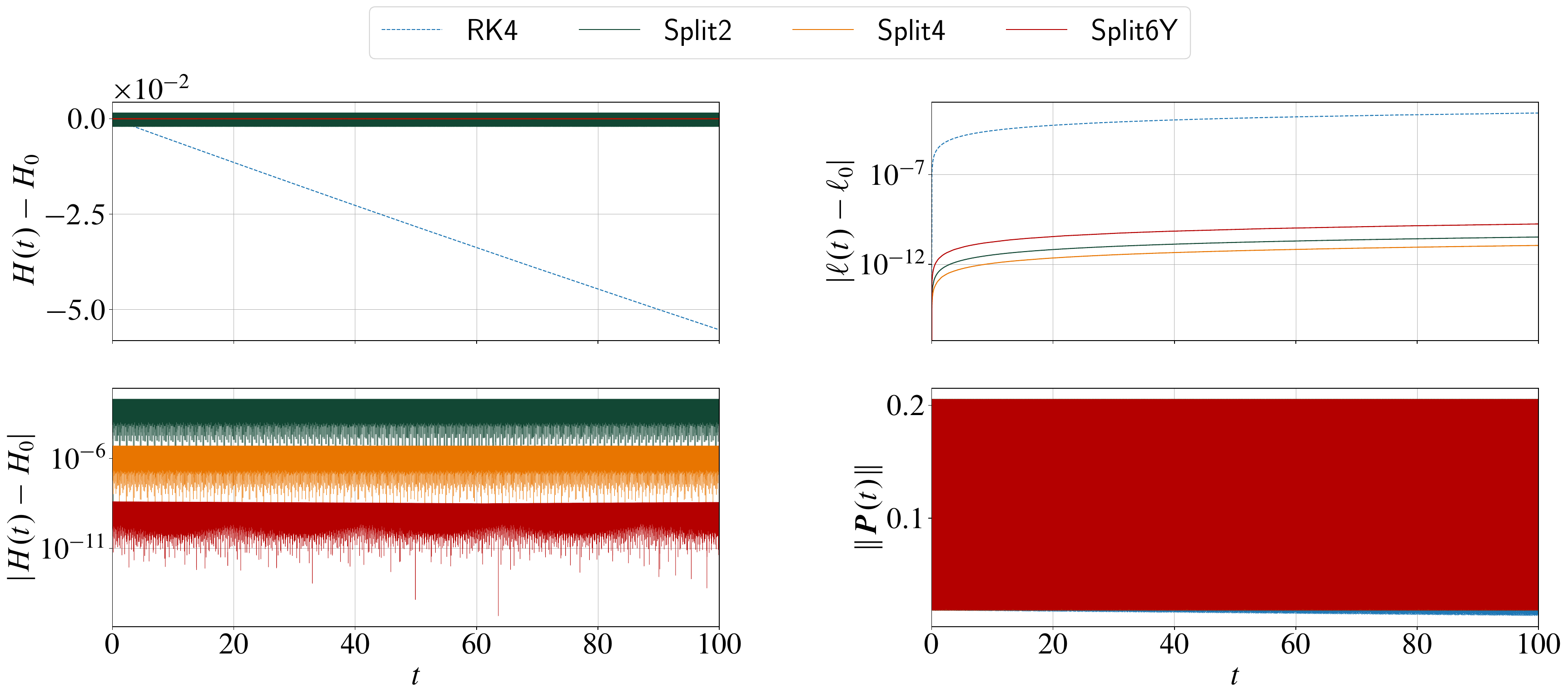}
  \caption{Time evolution of relative errors in Hamiltonian $H$, angular momentum $\ell$, and deviation $\norm{\bm{P}(t)}$ from $\mathcal{K}$ for the pentagonal necklace problem from \Cref{fig:x-y_necklace}.}
  \label{fig:errors_necklace}
\end{figure}

\section*{Summary and Outlook}
We have developed explicit integrators for the Hamiltonian dynamics~\eqref{eq:Hamilton} of massive point vortices that preserve the symplectic structure~\eqref{eq:Omega} and the angular momentum~\eqref{eq:ell} exactly, as well as nearly preserve the Hamiltonian~\eqref{eq:H} without drift.
Thanks to the preservation of these key invariants, the solutions exhibit excellent long-time accuracies compared to the Runge--Kutta method.
In particular, in the small-mass regime $\eps \ll 1$ of our interest here, the difference in accuracy is pronounced when the solutions become highly oscillatory.

Such a long-time accuracy and preservation of invariants are particularly important in numerically analyzing the stability of the massive vortices.
Given a recent interest in analyzing the stability of massive point vortices~\cite{PhysRevE.111.034216}, those symplectic integrators for massive point vortex dynamics in BEC with long-time accuracy will play an important role in numerically predicting the stability of massive vortices.

It is interesting to consider an extension of our integrators to other models of massive vortex dynamics, such as those presented in \cite{BeRiPe2023,PhysRevA.106.063307,PhysRevE.111.034216}, which seem to improve upon the model \eqref{eq:Hamilton} considered here.

\section*{Acknowledgments}
This work was supported by NSF grant DMS-2006736.
I would like to thank Andrea Richaud for introducing me to the subject of massive point vortices.

\bibliography{Massive_Vortex_Integrator}
\bibliographystyle{elsarticle-num}

\end{document}